\documentclass[a4paper,UKenglish]{article}

\usepackage{microtype}

\usepackage{complexity}
\usepackage{comment}
\usepackage{color}
\usepackage{footnote}
\usepackage[top=3cm,left=3cm,right=3cm,bottom=3cm]{geometry}
\usepackage[intlimits]{amsmath}
\usepackage{amsfonts,amssymb,amsthm,bbm, mathtools}
\usepackage{verbatim}
\usepackage{hyperref}
\usepackage{enumerate}
\usepackage[linesnumbered,ruled]{algorithm2e}

\hypersetup{
colorlinks = true,
citecolor = blue
}

\newtheorem{theorem}{Theorem}
\newtheorem{lemma}[theorem]{Lemma}
\newtheorem{fact}[theorem]{Fact}
\newtheorem{corollary}[theorem]{Corollary}
\theoremstyle{definition}
\newtheorem{definition}{Definition}

\newtheorem{claim}[theorem]{Claim}

\newclass{\alt}{alt}
\newclass{\s}{s}
\newclass{\bs}{bs}
\newclass{\fbs}{fbs}
\newclass{\FC}{FC}
\newclass{\Cert}{C}
\newclass{\EC}{EC}
\newclass{\RC}{RC}
\newclass{\SB}{sb}
\newclass{\fsb}{fsb}
\newclass{\sC}{sC}
\newclass{\cadv}{cadv}
\newclass{\prt}{prt}
\newclass{\corr}{corr}
\newclass{\Q}{Q}
\newclass{\Adv}{Adv}

\newcommand\delim{\,\middle|\,}
\newcommand{\clA}{\mathcal{A}}
\newcommand{\clB}{\mathcal{B}}
\newcommand{\clE}{\mathcal{E}}

\newcommand{\bbE}{\mathbb{E}}

\makeatletter
\newtheorem*{rep@theorem}{\textbf{\rep@title}}
\newcommand{\newreptheorem}[2]{%
\newenvironment{rep#1}[1]{%
 \def\rep@title{#2 \ref{##1}}%
 \begin{rep@theorem}}%
 {\end{rep@theorem}}}
\makeatother
\newreptheorem{theorem}{Theorem}

\newcommand{\email}[1]{\href{mailto:#1}{#1}}

\newcommand{\calB}{\mathcal{B}}

\title{Quadratically Tight Relations for Randomized Query Complexity}

\title{Quadratically Tight Relations \\ for Randomized Query Complexity}
\author{Dmitry Gavinsky \thanks{Institute of Mathematics, Czech Academy of Sciences, 115 67 \v Zitna 25, Praha 1, Czech Republic.} \and
Rahul Jain \thanks{Centre for Quantum Technologies, National University of Singapore, Block S15, 3 Science Drive 2, Singapore 117543. \email{rahul@comp.nus.edu.sg},  \email{cqthk@nus.edu.sg}, \email{srijita.kundu@u.nus.edu}.} \thanks{MajuLab, UMI 3654, Singapore.} \and 
Hartmut Klauck \footnotemark[2] \footnotemark[4] \and
Srijita Kundu   \footnotemark[2] \and 
Troy Lee \footnotemark[2] \thanks{SPMS, Nanyang Technological University, 21 Nanyang Link, Singapore 637371. \email{troyjlee@gmail.com}, \email{ssanyal@ntu.edu.sg}.} \and
Miklos Santha \footnotemark[2] \thanks{IRIF, Universit\'e Paris Diderot, CNRS, 75205 Paris, France. \email{santha@irif.fr}.
}\and
Swagato Sanyal \footnotemark[2] \footnotemark[4] \and
Jevg\={e}nijs Vihrovs \thanks{Centre for Quantum Computer Science, University of Latvia, Rai\c{n}a 19, Riga, Latvia, LV-1586. \email{jevgenijs.vihrovs@lu.lv}.}}
\date{}

\begin{document}

\maketitle

\begin{abstract}
Let $f:\{0,1\}^n \rightarrow \{0,1\}$ be a Boolean function. The certificate complexity $\C(f)$ is a complexity measure that is quadratically tight for the zero-error randomized query complexity $\R_0(f)$: $\C(f) \leq \R_0(f) \leq \C(f)^2$.
In this paper we study a new complexity measure that we call expectational certificate complexity $\EC(f)$, which is also a quadratically tight bound on $\R_0(f)$: $\EC(f) \leq \R_0(f) = O(\EC(f)^2)$.
We prove that $\EC(f) \leq \C(f) \leq \EC(f)^2$ and show that there is a quadratic separation between the two, thus $\EC(f)$ gives a tighter upper bound for $\R_0(f)$.
The measure is also related to the fractional certificate complexity $\FC(f)$ as follows: $\FC(f) \leq \EC(f) = O(\FC(f)^{3/2})$.
This also connects to an open question by Aaronson whether $\FC(f)$ is a quadratically tight bound for $\R_0(f)$, as $\EC(f)$ is in fact a relaxation of $\FC(f)$.

In the second part of the work, we upper bound the distributed query complexity $\D^\mu_\epsilon(f)$ for product distributions  $\mu$ by the square of the query corruption bound ($\corr_\epsilon(f)$) which improves upon a result of Harsha, Jain and Radhakrishnan [2015]. A similar statement for communication complexity is open. 

\end{abstract}

\section{Introduction}
\label{introduction}
The query model is arguably the simplest model for computation of Boolean functions. Its simplicity is convenient for showing lower bounds for the amount of time required to accomplish a computational task. In this model, an algorithm computing a function $f:\{0,1\}^n \rightarrow \{0,1\}$ on $n$ bits is given query access to the input $x \in \{0,1\}^n$. The algorithm can \emph{query} different bits of $x$, possibly in an adaptive fashion, and finally produces an output. The complexity of the algorithm is the number of queries made; in particular, the algorithm does not incur additional cost for any computation other than the queries.

Unlike the more general models of computation (e.g. Boolean circuits, Turing machines), it is often possible to completely determine the query complexity of explicit functions using existing tools and techniques. The study of query algorithms can thus be a natural first step towards understanding the computational power and limitations of more general and complex models. Query complexity has seen a long line of research by computational complexity theorists. We refer the reader to the survey by Buhrman and de Wolf \cite{Buhrman_deWolf_2002} for a comprehensive introduction to this line of work.

To understand query algorithms, researchers have defined many complexity measures of Boolean functions and investigated their relationship to query complexity, and to one another. For a summary of the current state of knowledge about these measures, see \cite{Aaronson_2016}. 
In this work, we focus on characterizing the bounded-error query complexity $\R(f)$ and the zero-error query complexity $\R_0(f)$.

The following measures are known to lower bound $\R_0(f)$: block sensitivity $\bs(f)$, fractional certificate complexity $\FC(f)$ (also known as fractional block sensitivity $\fbs(f)$, \cite{Tal_2013}), and certificate complexity $\C(f)$. They are related as follows:
\begin{equation*}
\bs(f) \leq \fbs(f) = \FC(f) \leq \C(f).
\end{equation*}
It is known that $\R_0(f) \leq \D(f) \leq \C(f)^2$, and the \textsc{Tribes} function (an \textsc{And} of $\sqrt n$ \textsc{Or}s on $\sqrt n$ bits) demonstrates that this relation is tight \cite{Jain_2010}. It is also known that $\R_0(f)=O(\bs(f)^3)=O(\FC(f)^3)$ \cite{Nisan_1989, Beals_2001}. A quadratic separation between $\R_0(f)$ and $\FC(f)$ is also achieved by \textsc{Tribes}. Aaronson posed a question whether $\R_0(f)=O(\FC^2(f))$ holds \cite{Aaronson_2008} (stated in terms of the randomized certificate complexity $\RC(f)$, which later has been shown to be equivalent to $\FC(f)$ \cite{Gilmer_2016}).
A positive answer to this question would imply that $\R_0(f) = O(\widetilde{\deg}(f)^4) = O(\Q(f)^4)$ \cite{Aaronson_2016}, where $\widetilde{\deg}(\cdot)$ and $\Q(\cdot)$ stand for approximate polynomial degree and quantum query complexity respectively.

One approach to showing $\R_0(f) \leq \FC(f)^2$ is to consider the natural generalization of the proof $\D(f) \leq \C(f)^2$ 
to the randomized case; the analysis of this algorithm, however, has met some unresolved obstacles \cite{Kulkarni_2016}.
We define a new complexity measure \emph{expectational certificate complexity} $\EC(f)$ that is specifically designed to avert these problems and is of a similar form to $\FC(f)$. We show that $\EC$ gives a quadratically tight bound for $\R_0$:
\begin{theorem} \label{ec-quad} For all total Boolean functions $f$, 
\begin{equation*}\EC(f) \leq \R_0(f) \leq O(\EC(f)^2).\end{equation*}
\end{theorem}
\noindent In fact, $\FC(f)$ is a relaxation of $\EC(f)$, and we show that $\FC(f) \leq \EC(f) \leq \C(f).$
Moreover, we show that $\EC(f)$ lies closer to $\FC(f)$ than $\C(f)$ does:
$\FC(f) \leq \EC(f) \leq \FC(f)^{3/2}.$
While we don't know whether $\EC(f)$ is a lower bound on $\R(f)$, the last property gives $\EC(f)^{2/3} \leq \R(f)$.

As mentioned earlier, $\C(f)^2$ bounds $\R_0(f)$ from above. But for specific functions, $\EC(f)^2$ can be an asymptotically tighter upper bound than $\C(f)^2$. We demonstrate that by showing that the same example that provides a quadratic separation between $\C(f)$ and $\FC(f)$ \cite{Gilmer_2016} also gives $\C(f)=\Omega(\EC(f)^2)$. This is the widest separation possible between $\EC(f)$ and $\C(f)$, because $\C(f) \leq \R_0(f) =O(\EC(f)^2)$.

In the second part of the paper, we upper bound the distributional query complexity for product distributions in terms of the minimum product query corruption bound and the block sensitivity (see Definition~\ref{corruption} and Section~\ref{preliminaries}).
\begin{theorem}
\label{prt_ub:this1}
Let $\epsilon \in [0,1/2)$ and $\mu$ a product distribution over the inputs.
Then
\begin{equation*}\D_{4\epsilon}^{\mu}(f) = O(\corr_{\min,\epsilon}^{\times}(f) \cdot \bs(f)).\end{equation*}
\end{theorem}
We contrast Theorem~\ref{prt_ub:this1} with the past work by Harsha, Jain and Radhakrishnan \cite{Harsha+_2015}, who showed that for product distributions, the distributional query complexity is bounded above by the square of the smooth corruption bound corresponding to inverse polynomial error. Theorem~\ref{prt_ub:this1} improves upon their result, firstly by upper bounding the distributional complexity by minimum query corruption bound, which is an asymptotically smaller measure than the smooth corruption bound, and secondly by losing a constant factor in the error as opposed to a polynomial worsening in their work. Theorem~\ref{D&corr}, a consequence of Theorem~\ref{prt_ub:this1},  shows that for product distribution over the inputs, the distributional query complexity is asymptotically bounded above by the square of the query corruption bound.
Thus Theorem~\ref{D&corr} resolves a question that was open after the work of Harsha et. al. The analogous question in communication complexity is still open.

Theorem~\ref{prt_ub:this1} also bounds distributional query complexity in terms of the \emph{partition bound} $\prt(\cdot)$ of Jain and Klauck \cite{Jain_2010}. The following theorem follows from Theorems~\ref{prt_ub:this1} and~\ref{ceps&prt}.
\begin{theorem}\label{prtjk}
If $\epsilon \in \left[0, \frac{1}{8}\right]$ then $\D_{8\epsilon}^{\mu}(f) = O(\prt_\epsilon(f)^2\cdot \log(1/\epsilon)).$
\end{theorem}
Jain and Klauck showed that $\prt(f)$ is a powerful lower bound on $\R(f)$. In the same work, $\prt(f)$ was used to give a tight $\Omega(n)$ lower bound on $\R(f)$ for the \textsc{Tribes} function on $n$ bits. The authors proved that $\prt(f)$ is asymptotically larger than $\FC(f)$. This implies that $\R(f)=O(\prt(f)^3)$, since $\R(f)=O(\bs(f)^3)$. While a quadratic separation between $\R(f)$ and $\prt(f)$ is known \cite{Ambainis_2016}, it is open whether $\R(f)=O(\prt(f)^2)$. Theorem~\ref{prt} proves a distributional version of this quadratic relation, for the special case in which the input is sampled from a product distribution. We remark here that Jain, Harsha and Radhakrishnan proved in their work that $\D_{1/3}^{\mu}(f) = O(\prt_{1/3}(f)^2\cdot (\log \prt_{1/3}(f))^2)$; Theorem~\ref{prt} achieves polylogarithmic improvement over this bound.
We note here that an analogous statement for an arbitrary distribution together with the Minimax Principle (see Fact~\ref{minimax}) will imply that $\R(f)=O(\prt(f)^2)$.

The paper is organized as follows.
In Section \ref{preliminaries}, we give the definitions for some of the complexity measures.
In Section \ref{ec}, we define the expectational certificate complexity and prove the results concerning this measure, starting with Theorem \ref{ec-quad}.
In Section \ref{prt}, we define the minimum query corruption bound and prove Theorems  \ref{prt_ub:this1} and \ref{prtjk}.
In Section \ref{open}, we list some open problems concerning our measures.

\section{Preliminaries}
\label{preliminaries}

In this section we recall the definitions of some known complexity measures.
For detailed introduction on the query model, see the survey \cite{Buhrman_deWolf_2002}.
For the rest of this paper, $f$ is any total Boolean function on $n$ bits, $f : \{0, 1\}^n \to \{0, 1\}$.

\begin{definition}[Randomized Query Complexity]
Let $\clA$ be a randomized algorithm that as an input takes $x \in \{0, 1\}^n$ and returns a Boolean value $\clA(x, r)$, where $r$ is any random string used by $\clA$.
With one query $\clA$ can ask the value of any input variable $x_i$, for $i \in [n]$.
The complexity $C(\clA, x, r)$ of $\clA$ on $x$ is the number of queries the algorithm performs under randomness $r$, given $x$.
The worst-case complexity of $\clA$ is $C(\clA) = \max_{r, x \in \{0, 1\}^n} C(\clA, x, r)$.

The \emph{zero-error randomized query complexity} $\R_0(f)$ is defined as  $\min_{\clA}\max_x \bbE_r [C(\clA, x, r)]$, where $\clA$ is any randomized algorithm such that for all $x \in \{0, 1\}^n$, we have $\Pr_r[\clA(x,r) = f(x)] = 1$.

The \emph{one-sided error randomized query complexity} $\R^0_{\epsilon}(f)$ is defined as $\min_{\clA}C(\clA)$, where $\clA$ is any randomized algorithm such that for every $x$ such that $f(x) = 0$, we have $\Pr_r[\clA(x,r) = 1] \leq \epsilon$, and for all $x$ such that $f(x) = 1$, we have $\Pr_r[\clA(x,r) = 1] = 1$.
Similarly we define $\R^1_{\epsilon}(f)$.

The \emph{two-sided error randomized query complexity} $\R_{\epsilon}(f)$ is defined as $\min_{\clA}C(\clA)$, where $\clA$ is any randomized algorithm such that for every $x \in \{0, 1\}^n$, we have $\Pr_r[\clA(x,r) \neq f(x)] \leq \epsilon$.
We denote $\R_{1/3}(f)$ simply by $\R(f)$.
\end{definition}

\begin{definition}[Distributional Query Complexity]
Let $\mu$ be a probability distribution over $\{0, 1\}^n$, and $\epsilon \in [0,1/2)$.
The \emph{distributional query complexity} $\D^\mu_\epsilon(f)$ is the minimum number of queries made in the worst case (over inputs) by a deterministic query algorithm $\clA$ for which $\Pr_{x \sim \mu}[\clA(x) = f(x)] \geq 1 - \epsilon$.
\end{definition}

The \emph{Minimax Principle} relates the randomized query complexity and distributional query complexity measures of Boolean functions.
\begin{fact}[Minimax Principle]\label{minimax}
For any Boolean function $f, \R_\epsilon(f)=\max_{\mu}\D_\epsilon^\mu(f).$
\end{fact}

\begin{definition}[Product Distribution]
A probability distribution $\mu$ over $\{0, 1\}^n$ is a \emph{product distribution} if there exist $n$ functions $\mu_1, \ldots, \mu_n : \{0, 1\} \to [0, 1]$ such that $\mu_i(0) + \mu_i(1) = 1$ for all $i$ and for all $x \in \{0, 1\}^n$,
\begin{equation*}\mu(x) = \prod_{i \in [n]} \mu_i(x_i).\end{equation*}
\end{definition}

\begin{definition}[Certificate Complexity]
An \emph{assignment} is a map $A:\{1, \ldots, n\} \rightarrow \{0,1, *\}$.
All inputs consistent with $A$ form a subcube $\{x \in \{0,1\}^n \mid \forall i \in [n]: x_i=A(i) \text{ or } A(i) = *\}$.
The length or size of an assignment, denoted by $|A|$, is defined to be the co-dimension of the subcube it corresponds to.
Let $Q_A:=\{j:A(j) \neq *\}$ be the set of variables fixed by $A$.

For $b \in \{0, 1\}$, a $b$-certificate for $f$ is an assignment $A$ such that $x \in A \Rightarrow f(x) = b$.
The \emph{certificate complexity} $\C(f, x)$ of $f$ on $x$ is the size of the shortest $f(x)$-certificate that is consistent with $x$.
The certificate complexity of $f$ is defined as $\C(f) = \max_{x \in \{0, 1\}^n} \C(f, x)$.
The $b$-certificate complexity of $f$ is defined as $\C^b(f) = \max_{x : f^{-1}(b)} \C(f, x)$.
\end{definition}

\begin{definition}[Sensitivity and Block Sensitivity]
For $x \in \{0, 1\}^n$ and $S \subseteq [n]$, let $x^S$ be $x$ flipped on locations in $S$.
The \emph{sensitivity} $\s(f, x)$ of $f$ on $x$ is the number of different $i \in [n]$ such that $f(x) \neq f(x^{\{i\}})$.
The sensitivity of $f$ is defined as $\s(f) = \max_{x \in \{0, 1\}^n} \s(f, x)$.

The \emph{block sensitivity} $\bs(f, x)$ of $f$ on $x$ is the maximum number $k$ of disjoint subsets $B_1, \ldots, B_k \subseteq [n]$ such that $f(x) \neq f(x^{B_i})$ for each $i \in [k]$.
The block sensitivity of $f$ is defined as $\bs(f) = \max_{x \in \{0, 1\}^n} \bs(f, x)$.
\end{definition}

\begin{definition}[Fractional Certificate Complexity]
The \emph{fractional certificate complexity} $\FC(f, x)$ of $f$ on $x \in \{0, 1\}^n$ is defined as the optimal value of the following linear program:
\begin{equation*}
\text{minimize } \sum_{i \in [n]} v_x(i) \hspace{1cm} \text{ subject to } \quad \forall y \text{ s.t. } f(x) \neq f(y): \sum_{i : x_i \neq y_i} v_x(i) \geq 1.
\end{equation*}
Here $v_x \in \mathbb R^n$ and $v_x(i) \geq 0$ for each $x \in \{0, 1\}^n$ and $i \in [n]$.
The fractional certificate complexity of $f$ is defined as $\FC(f) = \max_{x \in \{0, 1\}^n} \FC(f, x)$.
\end{definition}

\begin{definition}[Fractional Block Sensitivity]
Let $\clB =\{B \mid f(x) \neq f(x^B)\}$ be the set of sensitive blocks of $x$.
The \emph{fractional block sensitivity} $\fbs(f, x)$ of $f$ on $x$ is defined as the optimal value of the following linear program:
\begin{align*}
\text{maximize } \sum_{B \in \clB} u_x(y) \hspace{1cm} \text{ subject to } \quad & \forall i \in [n]: \sum_{B \in \clB \atop i \in B} u_x(B) \leq 1.
\end{align*}
Here $u_x \in \mathbb R^{|\clB|}$ and $u_x(B) \leq 1$ for each $x \in \{0, 1\}^n$ and $B \in \clB$.
The fractional block sensitivity of $f$ is defined as $\fbs(f) = \max_{x \in \{0, 1\}^n} \fbs(f, x)$.
\end{definition}

The linear programs $\FC(f, x)$ and $\fbs(f, x)$ are duals of each other, hence their optimal solutions are equal and $\FC(f) = \fbs(f)$ \cite{Gilmer_2016}.

\section{Expectational Certificate Complexity} \label{ec}
In this section, we give the results for the expectational certificate complexity.
The measure is motivated by the well-known $\D(f) \leq \C^0(f)\C^1(f)$ deterministic query algorithm which was
independently discovered several times \cite{Blum_1987,HH87,T90}.
In each iteration, the algorithm queries the set of variables fixed by some consistent 1-certificate.  Either the query 
answers agree with the fixed values of the 1-certificate, in which case the input must evaluate to 1, or the algorithm makes
progress as the 0-certificate complexity of all 0-inputs still consistent with the query answers is decreased by at least 1.
The latter property is due to the crucial fact that the set of fixed values of any 0-certificate and 1-certificate must intersect.

In hopes of proving $\R(f) \leq \FC^0(f)\FC^1(f)$, a straightforward generalization to a randomized algorithm would be to 
pick a consistent 1-input $x$ and query each variable independently with probability $v_x(i)$, where $v_x$ is a fractional 
certificate for $x$.  To show that such an algorithm makes progress, one needs a property analogous to the fact that 
0-certificates and 1-certificates overlap. Kulkarni and Tal give a similar intersection property for the fractional certificates:
\begin{lemma}[\cite{Kulkarni_2016}, Lemma 6.2] \label{KulkarniTal}
Let $f : \{0, 1\}^n \to \{0, 1\}$ be a total Boolean function and $\{v_x\}_{x \in \{0, 1\}^n}$ be an optimal solution for the $\FC(f)$ linear program.
Then for any two inputs $x, y \in \{0, 1\}^n$ such that $f(x) \neq f(y)$, we have
\begin{equation*}
\sum_{i : x_i \neq y_i} \min\{v_x(i), v_y(i)\} \geq 1.
\end{equation*}
\end{lemma}
\noindent However, it is not clear whether the algorithm makes progress in terms of reducing the fractional certificates of the 0-inputs.
We get around this problem by replacing $\min\{v_x(i), v_y(i)\}$ with the product $v_x(i)v_y(i)$ and putting that the 
sum of these terms over $i$ where $x_i \ne y_i$ is at least 1 as a constraint:

\begin{definition}[Expectational Certificate Complexity]
The \emph{expectational certificate complexity} $\EC(f)$ of $f$ is defined as the optimal value of the following program:
\begin{align*}
\text{minimize } \max_x \sum_{i=1}^n w_x(i) \hspace{.5cm} \text{s.t.} \quad &\sum_{i : x_i \neq y_i} w_x(i) w_y(i) \geq 1 \text{ for all $x, y$ s.t. $f(x) \neq f(y)$}, \\
& 0 \leq w_x(i) \leq 1 \text{ for all $x \in \{0, 1\}^n$, $i \in [n]$.}
\end{align*}
\end{definition}

\noindent We use the term ``expectational'' because the described algorithm on expectation queries at least weight 1 in total from input $y$, when querying the variables with probabilities being the weights of $x$.
While the informally described algorithm shows a quadratic upper bound on the worst-case expected complexity, in the next section we show a slight modification that directly makes a quadratic number of queries in the worst case.

\subsection{Quadratic Upper Bound on Randomized Query Complexity}

In this section we prove Theorem~\ref{ec-quad} (restated below).
\begin{reptheorem}{ec-quad}
\[\EC(f) \leq \R_0(f) \leq O(\EC(f)^2).\]
\end{reptheorem}
\begin{proof}
The first inequality follows from Lemma~\ref{ec-c} and $\C(f) \leq \R_0(f)$.

To prove the second inequality, we give randomized query algorithms for $f$ with 1-sided error $\epsilon$.
\begin{claim}
\label{onesided}
For any $b \in \{0, 1\}$, we have $\R_{\epsilon}^b(f) \leq \lceil\EC(f)^2/\epsilon\rceil$.
\end{claim}
The second inequality of Theorem~\ref{ec-quad} follows from Claim~\ref{onesided} by standard arguments of $\ZPP=\RP \cap \coRP$.
\begin{proof}[Proof of Claim~\ref{onesided}]
We prove the claim for $b=0$. The case $b=1$ is similar.

Let $\{w_x\}_{x \in \{0,1\}^n}$ be an optimal solution to the $\EC(f)$ program.
We say that an input $y$ is consistent with the queries made by $\clA$ on $x$ if $y_i=x_i$ for all queries $i \in [n]$ that 
have been made.
Also define a probability distribution $\mu_y(i) = w_y(i)/\sum_{i \in [n]} w_y(i)$ for each input $y \in \{0, 1\}^n$.

\begin{algorithm}[H]
\KwIn{$x \in \{0, 1\}^n$}

\begin{enumerate}
\item Repeat $\lceil \EC(f)^2/\epsilon \rceil$ many times:
\begin{enumerate}
\item Pick the lexicographically first consistent 1-input $y$.
If there is no such $y$, return 0.
\item Sample a position $i$ from $\mu_y$ and query $x_i$.
\item If the queried values form a $c$-certificate, return $c$.
\end{enumerate}
\item \label{end} Return 1.
\end{enumerate}

\caption{The randomized query algorithm $\clA$.}
\end{algorithm}

The complexity bound is clear as $\clA$ always performs at most $\lceil \EC(f)^2/\epsilon \rceil$ queries.

For correctness, note that the algorithm outputs 1 on all 1-inputs.
Thus assume $x$ is a 0-input from here on in the analysis.
Then we have to prove that $\clA$ outputs 0 with probability at least $1-\epsilon$.
This amounts to showing that the function reduces to a constant 0 function and the algorithm terminates within  
$\lceil \EC(f)^2/\epsilon \rceil$ iterations with probability at least $1-\epsilon$.  (For notational convenience, in what follows 
we will drop the ceilings and assume $\EC(f)^2/\epsilon$ is an integer.)

Define a random variable $T_k$ as
\begin{equation*}
T_k =
\begin{cases}
\frac{1}{\EC(f)}, & \text{if $\clA$ has terminated before the $k$-th iteration,} \\
w_x(i), & \text{if at the $k$-th iteration $\clA$ has queried $x_i$ for the first time,} \\
0, & \text{if $x_i$ has been queried before the $k$-th iteration.}
\end{cases}
\end{equation*}
Let $T = \sum_{k=1}^{\EC(f)^2/\epsilon} T_k$.
As $\sum_{i \in [n]} w_x(i) \leq \EC(f)$ by definition, $T > \EC(f)$ implies that $\clA$ has terminated before point \ref{end}.
Then it has returned 0, and the answer is correct.
Let $p = \Pr[T > \EC(f)]$.
We will prove that $p \geq 1-\epsilon$, in which case we would be done.

We continue by showing an upper and a lower bound on $\bbE[T]$.
\begin{itemize}
\item The maximum possible value of $T$ is at most
\begin{equation*}
T \leq \sum_{i \in [n]} w_x(i) + \frac{\EC(f)^2}{\epsilon} \cdot \frac{1}{\EC(f)} \leq \left(1 + \frac{1}{\epsilon}\right) \EC(f).
\end{equation*}
Therefore,
\begin{equation*}
\bbE[T] \leq p \cdot \left(1 + \frac{1}{\epsilon}\right) \EC(f) + (1-p) \cdot \EC(f) \leq \left(1+\frac{p}{\epsilon}\right) \EC(f).
\end{equation*}
\item Let $\clE_k$ be the event that $\clA$ has terminated before the $k$-th iteration.
In case $\clA$ performs the $k$-th iteration, let $y$ be consistent 1-input chosen and the random variable $i_k$ be the position that $\clA$ queries.
\begin{align*}
\bbE[T_k] &= \Pr[\clE_k] \cdot \frac{1}{\EC(f)} +  \Pr[\overline{\clE_k}] \cdot \bbE[w_x(i_k) \mid \overline{\clE_k}] \\
&\geq \Pr[\clE_k] \cdot \frac{1}{\EC(f)} + \Pr[\overline{\clE_k}] \cdot \sum_{i : x_i \neq y_i} w_x(i) \mu_y(i) \\
&= \Pr[\clE_k] \cdot \frac{1}{\EC(f)} + \Pr[\overline{\clE_k}] \cdot \sum_{i : x_i \neq y_i} w_x(i) w_y(i) / \sum_{i \in [n]} w_y(i)\\
&\geq \Pr[\clE_k] \cdot \frac{1}{\EC(f)} + \Pr[\overline{\clE_k}] \cdot \frac{1}{\EC(f)} \\
&= \frac{1}{\EC(f)},
\end{align*}
The first inequality here follows from the fact that any $i$ such that $x_i \neq y_i$ has not been queried yet, because $x$ and $y$ are both consistent with the queries made so far.
Thus, the inequality holds regardless of the randomness chosen by $\clA$.
The second inequality follows from the expectational certificate properties $\sum_{i:x_i\neq y_i} w_x(i)w_y(i) \geq 1$ and $\sum_{i \in [n]} w_y(i) \leq \EC(f)$.
By the linearity of expectation, we have that
\begin{equation*}
\bbE[T] = \sum_{k = 1}^{\EC(f)^2/\epsilon} \bbE[T_k] \geq \EC(f)/\epsilon.
\end{equation*}
\end{itemize}
Combining the two bounds together, we get $\frac{\EC(f)}{\epsilon} \leq \left(1+\frac{p}{\epsilon}\right) \EC(f).$
Thus, $p \geq 1-\epsilon$.
\end{proof}
\end{proof}

\subsection{Relation with the Fractional Certificate Complexity}

\begin{lemma}
$\FC(f) \leq \EC(f).$
\end{lemma}

\begin{proof}
We show that a feasible solution $\{w_x\}_x$ for $\EC(f)$ is also feasible for $\FC(f)$.
Since $0 \leq w_x(i) \leq 1$ for any $x, i$,
\begin{equation*}
\sum_{i : x_i \neq y_i} w_x(i) \geq \sum_{i : x_i \neq y_i} w_x(i)w_y(i) \geq 1,
\end{equation*}
and we are done.
\end{proof}

\begin{lemma}
$\EC(f) = O(\FC(f) \sqrt{\s(f)}).$
\end{lemma}

\begin{proof}
Let $\{v_x\}_x$ be an optimal solution to the fractional certificate linear program for $f$.
We first modify each $v_x$ to a new feasible solution $v'_x$ by eliminating the entries $v_x(i)$ that are very small, and boosting the large entries by a 
constant factor.
Namely, let
\begin{equation*}
v'_x(i) = \begin{cases} \min\left\{\frac{3}{2}v_x(i), 1\right\}, & \text{if } v_x(i) \geq \frac{1}{3s(f)}, \\ 0, & \text{otherwise.} \end{cases}
\end{equation*}
 We first claim that $\{v'_x\}_x$ is still a feasible solution.  Fix any $x \in \{0,1\}^n$, and let $B$ be a minimal sensitive block for $x$.
 As $v_x$ is part of a feasible solution, we have
 \begin{equation*}
 1 \le \sum_{i \in B} v_x(i) = \sum_{i \in B, \atop v_x(i) < 1/3s(f)} v_x(i) + \sum_{i \in B, \atop v_x(i) \ge 1/3s(f)} v_x(i) 
 \le \frac{1}{3} +  \sum_{i \in B, \atop v_x(i) \ge 1/3s(f)} v_x(i) .
 \end{equation*}
 The second line follows because $|B| \le s(f)$, as $B$ is a minimal sensitive block and therefore every index in $B$ is sensitive.
 Rearranging the last inequality, we have
$ \sum_{i \in B \atop v_x(i) \ge 1/3s(f)} v_x(i) \ge \frac{2}{3},$
and therefore,
$ \sum_{i \in B} v'_x(i) \ge 1.$

Next, $w_x(i) := \sqrt{v'_x(i)}$ is a feasible solution to the expectational certificate program, as
\begin{equation*}
\sum_{i : x_i \neq y_i} w_x(i)w_y(i) = \sum_{i: x_i \neq y_i}\sqrt{v'_x(i)v'_y(i)} \geq \sum_{i: x_i \neq y_i}\min\{v'_x(i), v'_y(i)\} \geq 1.
\end{equation*}
 The second inequality holds by Lemma \ref{KulkarniTal}.
 
Now that we have shown that $\{w_x\}_x$ forms a feasible solution to the expectation certificate program, it 
remains to bound its objective value:
\begin{equation*}
\sum_{i\in [n]} w_x(i) = \sum_{i \in [n]}\sqrt{v'_x(i)} = \sum_{i: v'_x(i) \neq 0}\frac{v'_x(i)}{\sqrt{v'_x(i)}} \leq \sqrt{3s(f)}\sum_{i \in [n]} v'_x(i) \leq \sqrt{3s(f)}\frac{3}{2}\FC(f),\qedhere
\end{equation*}
where the first inequality follows from $v_x'(i) \geq v_x(i) \geq 1/3s(f)$ for $v_x'(i) \neq 0$.
\end{proof}

Since $\s(f) \leq \FC(f)$ and $\FC(f) \leq \R(f)$, we immediately get
\begin{corollary}
$\EC(f) = O(\FC(f)^{3/2}) = O(\R(f)^{3/2}).$
\end{corollary}

\subsection{Relation with the Certificate Complexity}

\begin{lemma} \label{ec-c}
$\EC(f) \leq \C(f).$
\end{lemma}

\begin{proof}
We construct a feasible solution $\{w_x\}_x$ for $\EC(f)$ from $\C(f)$.
Let $A_x$ be the shortest certificate for $x$.
Assign $w_x(i) = 1$ iff $i \in A_x$, otherwise let $w_x(i) = 0$.
Let $x, y$ be any two inputs such that $f(x) \neq f(y)$.
There is a position $i$ where $A_x(i) \neq A_y(i)$, otherwise there would be an input consistent with both $A_x$ and $A_y$, which would give a contradiction.
Therefore, $w_x(i)w_y(i) \geq 1$.
The value of this solution is $\max_x \sum_{i\in [n]} w_x(i) = \max_x \C(f, x) = \C(f)$.
\end{proof}

As $\FC(f) \leq \EC(f) \leq \C(f) \leq \FC(f)^2$, there can be at most quadratic separation between $\EC(f)$ and $\C(f)$.
We show that this is achieved by the example of Gilmer et.~al. that separates $\FC(f)$ and $\C(f)$ quadratically:
\begin{theorem}[\cite{Gilmer_2016}, Theorem 32]
For every $n \in \mathbb N$ sufficiently large, there is a function $f : \{0, 1\}^{n^2} \to \{0, 1\}$ such that $\FC(f) = O(n)$ and $\C(f) = \Omega(n^2)$.
\end{theorem}
Their construction for $f$ is as follows.
First a function $g : \{0, 1\}^n \to \{0, 1\}$ is exhibited such that $\FC^0(g) = \Theta(1)$, $\C^0(g) = \Theta(n)$ and $\FC^1(g) = \C^1(f) = n$.
The function $f : \{0, 1\}^{n^2} \to \{0, 1\}$ is defined as a composition $\textsc{Or}(g(x^{(1)}), \ldots, g(x^{(n)}))$.
This gives $\FC(f) = \max\{n \FC^0(g), \FC^1(g)\} = \Theta(n)$ and $\C(f) \geq n \C^0(g) = \Theta(n^2)$ (both properties follow by Proposition 31 in their paper).

Let us construct a feasible solution $w$ for $\EC(f)$.
For any $x = x^{(1)} \dots x^{(n)}$ such that $f(x) = 1$, let $j$ be the first index such that $g(x^{(j)}) = 1$.
Let $S \subseteq [n^2]$ be the set of positions that correspond to $x^j$.
Let $w_x(i) = 1$ for each position $i$ in $S$, and $w_x(i) = 0$ for all other positions.
Then $\sum_{i=1}^{n^2} w_x(i) = n$.

On the other hand, let $\{v_x\}_{x \in \{0, 1\}^n}$ be an optimal solution to $\FC(f)$.
For any $x \in \{0, 1\}^{n^2}$ such that $f(x) = 0$, let $w_x(i) = v_x(i)$ for all $i \in [n^2]$.
Then $\sum_{i=1}^{n^2} w_x(i) = \FC(f, x) = O(n)$.

Now, for any two inputs $x, y$ such that $f(x) = 1$ and $f(y) = 0$, let $j$ be the smallest index such that $g(x^{(j)}) = 1$, then we have $g(y^{(j)}) = 0$.
By construction,
\begin{equation*}
\sum_{i : x_i \neq y_i} w_x(i) w_y(i) = \sum_{i : x_i \neq y_i} w_y(i) \geq 1.
\end{equation*}
Hence $\{w_x\}_x$ is a feasible solution to the expectational certificate and $\EC(f) = n$.

\section{Minimum Query Corruption Bound and Partition Bound} \label{prt}
In this section we prove Theorem~\ref{prt_ub:this1}. We first consider the \emph{query corruption bound} and \emph{minimum query corruption bound}.
\begin{definition}[Query Corruption Bound and Minimum Query Corruption Bound for product distributions]
\label{corruption}
Let $\epsilon \in [0,1/2)$ and $\mu : \{0,1\}^n \rightarrow [0,1]$ be a probability distribution over the inputs.
For a $b \in \{0, 1\}$, let an assignment $A$ be an $\epsilon$-error $b$-certificate under $\mu$, if
\begin{equation*}\Pr_{x \sim \mu}\left[f(x) \neq b \delim x \in A\right] \leq \epsilon.\end{equation*}
Define the \emph{query corruption bound} for $b$, distribution $\mu$ and error $\epsilon$ as
\begin{equation*}\corr_{\epsilon}^{b,\mu}(f) = \min\{|A| \mid \text{$A$ is an $\epsilon$-error $b$-certificate under $\mu$}\}.\end{equation*}
The query corruption bound of $f$ is defined as $\corr_\epsilon(f) = \max_{\mu} \max_b \corr_{\epsilon}^{b,\mu}(f)$, where $\mu$ ranges over all distributions on $\{0,1\}^n$.
The \emph{minimum} query corruption bound of $f$ \emph{for product distributions} is defined as $\corr^\times_{\min, \epsilon}(f) = \max_{\mu} \min_b \corr_{\epsilon}^{b,\mu}(f)$, where $\mu$ ranges over all product distributions on $\{0,1\}^n$.
\end{definition}

We now proceed to the proof of Theorem~\ref{prt_ub:this1} (restated below).
\begin{reptheorem}{prt_ub:this1}
Let $\epsilon \in [0,1/2)$ and $\mu$ a product distribution over the inputs.
Then
\begin{equation*}\D_{4\epsilon}^{\mu}(f) = O(\corr^\times_{\min, \epsilon}(f) \cdot \bs(f)).\end{equation*}
\end{reptheorem}

In the proof we will have restrictions of probability distributions.
Let $\eta$ be a probability distribution over $\{0,1\}^n$, $x \in \{0,1\}^n$ be a $n$-bit string, and $Q \subseteq \{1, \ldots, n\}$ be a set of indices. The restriction of $x$ to the indices of $Q$, $(x_j:j \in Q)$, will be denoted by $x_Q$. Then the distribution $\eta \mid_{x_Q}$ is the distribution obtained by conditioning $\eta$ on the event that the bits in the locations in $Q$ agree with $x$. Formally, for each $y \in \{0,1\}^n$
\[\eta_{x_Q}(y)=\left\{ \begin{array}{ll}  \frac{\eta(y)}{\sum_{z: \forall i \in Q, z_i=x_i}\eta(z)} & \mbox{if $\forall i \in Q, y_i=x_i$,} \\ 0 & \mbox{otherwise.}\end{array}\right.\]
\begin{proof}[Proof of Theorem~\ref{prt_ub:this1}]
We present a deterministic query algorithm, and analyse its performance for inputs sampled according to $\mu$. Examine the following algorithm:

\begin{algorithm}[H]
\KwIn{$x \in \{0, 1\}^n$}

\begin{enumerate}
\item Set $t_0 ,t_1 \gets 0, i \gets 1, \eta^{(1)} \gets \mu$.
\item Repeat:
\begin{enumerate}
\item \label{pick} Pick a shortest $\epsilon$-error certificate $A$ under $\eta$.
\item \label{query} Query all the variables in $Q_A$ that are still unknown.
\item Let $A$ be an $\epsilon$-error $b$-certificate for some $b \in \{0, 1\}$. Set $t_b \leftarrow t_b+1$.
\item \label{ret2} If the results of the queries are consistent with $A$, return $b$.
\item \label{ret3}If $t_b = 2\bs(f)$, return $b$.
\item $\eta^{(i+1)} \gets \eta^{(i)} \mid_{x_{Q_A}}$.
\item $i \gets i+1$.
\end{enumerate}
\end{enumerate}

\caption{The deterministic query algorithm $\mathcal B$.}
\end{algorithm}
For each $i=2, \ldots, 4 \bs(f)$, define $T^{(i)}$ to be the event that $\calB$ completes at least $i-1$ iterations and define $T^{(1)}$ to be the \emph{true} event. Let $i$ be arbitrary, and assume that $T^{(i)}$ occurs. Then $A^{(i)}$ denotes the $\epsilon$-error certificate (under $\eta^{(i)}$) picked in the $i$-th iteration in step~\ref{pick}. 
Let $b^{(i)} \in \{0, 1\}$ be the value approximately certified by $A^{(i)}$ under $\eta^{(i)}$.
Let $E^{(i)} \subseteq A^{(i)}$ denote the set of inputs $y \in A^{(i)}$ such that $f(y)\neq b^{(i)}$. 
Recall from Section~\ref{preliminaries} $Q_{A^{(i)}}$ is the set of variables set by $A^{(i)}$.
For each assignment $s \in \{0,1\}^{Q_{A^{(i)}}}$ to the variables fixed by $A^{(i)}$ and subset $U \subseteq A^{(i)}$, let $U\oplus s$ denote the shift of $U$ by the vector $s$. Formally ($\textquoteleft\oplus$' stands for bitwise exlusive or),
\[U\oplus s:=\{y \in \{0,1\}^n: \forall j \in Q_{A^{(i)}}, y_j=A^{(i)}_j\oplus s_j\mbox{ and }\exists z \in U\mbox{ such that } \forall j \notin Q_{A^{(i)}}, y_j=z_j\}.\]
For $i \geq 2$, define $\mathcal L^{(i)}$ to be the set of variables queried in first $i-1$ iterations and define $\mathcal L^{(1)}:=\varnothing$. Note that $\eta^{(i)}=\mu \mid_{x_{\mathcal L^{(i)}}}$, and $\eta^{(i)}$ is a product distribution.

Define all the above random variables to be $\bot$ if $T^{(i)}$ does not take place. Now define \[X^{(i)}=\left\{\begin{array}{ll}1 & \mbox{if $T^{(i)}$ occurs and $x \in \bigcup_{s \in \{0,1\}^{Q^{(i)}}} E^{(i)}\oplus s$,}\\ 0 & \mbox{otherwise.}\end{array}\right.\]
First we bound the number of queries made by $\calB$. Since $\clB$ terminates when either $t_0 = 2\bs(f)$ or $t_1 = 2\bs(f)$, it performs at most $4\bs(f)-1$ many iterations.
On the other hand since $\eta^{(i)}$ is a product distribution for each $i$, therefore $|A^{(i)}| \leq \corr^\times_{\min, \epsilon}(f)$. Therefore, the algorithm makes $O(\corr^\times_{\min, \epsilon}(f) \cdot \bs(f))$ many queries.

Now we prove that it errs on at most $4\epsilon$ fraction of the inputs according to $\mu$.
\begin{claim}
\label{noshift}
For every $i$ and $s \in \{0,1\}^{Q_{A}}, \Pr[x \in E^{(i)}\oplus s \mid T^{(i)}, x \in A^{(i)}\oplus s] \leq \epsilon$.
\end{claim}
\begin{proof}Condition on the events $T^{(i)},  x \in A^{(i)}\oplus s$. Furthermore, condition on $x_{\mathcal L^{(i)}}$. Notice that under this conditioning, the distribution of the input $x$ is $\eta^{(i)}=\mu \mid _{x_{\mathcal L^{(i)}}}$.

If $T^{(i)}$ occurs, $A^{(i)}$ is an $\epsilon$-error $b^{(i)}$-certificate under $\eta^{(i)}$. So $\Pr_{x \sim \eta^{(i)}}[x \in E^{(i)} \mid T^{(i)}, x \in A^{(i)}] \leq \epsilon$. Since $\eta^{(i)}$ is a product distribution as observed before, we have that for each $s \in \{0,1\}^{Q_{A^{(i)}}}, \Pr_{x \sim \eta^{(i)}}[x \in E^{(i)}\oplus s \mid T^{(i)}, x \in A^{(i)}\oplus s]=\Pr_{x \sim \eta^{(i)}}[x \in E^{(i)} \mid T^{(i)}, x \in A^{(i)}]\leq \epsilon$. The claim follows.
\end{proof}
In particular, Claim~\ref{noshift} implies that for all $i=1, \ldots, 4\bs(f)$,
\begin{equation}
\Pr[X^{(i)}=1] \leq \epsilon. \label{one}
\end{equation}
Since $\calB$ runs for at most $4\bs(f)-1 < 4\bs(f)$ steps, by Equation (\ref{one}), linearity of expectation and Markov's inequality we have that
\begin{align}
\Pr[|\{i \mid X^{(i)}=1\}| \geq \bs(f)] \leq 4\epsilon. \label{two}
\end{align}
For $i$ such that $T^{(i)}$ occurs, define $S^{(i)}:=\{j \in Q_{A^{(i)}} \mid x_j \neq A^{(i)}(j)\}$. The following claim will play a central role in our analysis.
\begin{claim}
\label{sensitive_block}
Let $i_1 < i_2$. For each $i \in \{i_1, i_2\}$, let $T^{(i)}$ happen and $X^{(i)}=0$. Then $f(x^{S^{(i)}})=b^{(i)}, S^{(i_1)} \cap S^{(i_2)}=\varnothing$. In particular, if $b^{(i_1)}=b^{(i_2)}$ and $f(x)=1-b^{(i_1)}$ then $S^{(i_1)}$ and $S^{(i_2)}$ are disjoint sensitive blocks for $x$. 
\end{claim}
\begin{proof}
Clearly, $x^{S^{(i)}} \in A^{(i)}$. Also, since $X^{(i)}=0, x \notin E^{(i)}+s$ for any $s$. Thus $x^{S^{(i)}} \notin E^{(i)}$. Hence $f(x^{S^{(i)}})=b^{(i)}$. To see that $S^{(i_1)} \cap S^{(i_2)}=\varnothing$, let $j \in S^{(i_1)}$.  It is easy to see that $i_2>i_1$ implies that the distribution $\eta^{(i_2)}$ at step $i_2$ is supported only on inputs consistent with $x_{Q_{A^{(i_1)}}}$. Hence, if $j \in Q_{A^{(i_2)}}$, then $x_j=A^{(i_2)}(j)$ which implies that $j \notin S^{(i_2)}$.
\end{proof}
For the rest of the proof, condition on the event that $\calB$ terminates at iteration $i$. We will bound the probability that $\calB$ errs.

First, condition on the event that $\calB$ terminates in step~\ref{ret2}. Then the probability that it errs is $\Pr[x \in E^{(i)} \mid T^{(i)}, x \in A^{(i)}] \leq \epsilon$ (by Claim~\ref{noshift} invoked with $s=0^{Q_{A^{(i)}}}$).

Next, condition on the event that $\calB$ terminates at step~\ref{ret3}, and $t_0 = 2\bs(f)$ (the case $t_1 = 2\bs(f)$ is symmetrical). By Equation~(\ref{two}), $|\{i \mid X^{(i)}=1\}| \geq \bs(f)$ with probability at most $4\epsilon$. Condition on $|\{i \mid X^{(i)}=1\}| < \bs(f)$. Then $\calB$ outputs $0$. We claim that $f(x)=0$ with probability $1$. Towards a contradiction, assume that $f(x)=1$. As $t_0 = 2\bs(f)$ and $|\{i \mid X^{(i)}=1\}| < \bs(f)$, then in at least $2\bs(f)-(\bs(f)-1)=\bs(f)+1$ iterations $j \leq i$, $b^{(j)}=0$ and $X^{(j)}=0$. By Claim~\ref{sensitive_block}, the blocks $S^{(j)}$ for those $j$ iterations are sensitive for $x$ and are disjoint. Since any input can have at most $\bs(f)$ sensitive blocks, we have the desired contradiction.

Thus the probability that $\calB$ errs is at most $\max\{\epsilon, 4\epsilon\}=4\epsilon.$ \qedhere
\end{proof}
Now we prove Theorem~\ref{prtjk}. Below we reproduce the definition of the partition bound by Jain and Klauck \cite{Jain_2010}. Here $\epsilon$ is an error parameter between 0 and 1, $A$ stands for subcubes, or equivalently, partial assignments, $z$ stands for a bit, i.e., a $0$ or a $1$, and $x$ stands for an input to $f$ from $\{0,1\}^n$.

\begin{definition}[Partition Bound] The $\epsilon$-\emph{partition bound} bound of $f$, denoted $\prt_{\epsilon}(f)$, is given by the logarithm of the optimal value of the following linear program\footnote{Jain and Klauck in their paper defined $\prt_\epsilon(f)$ to be the value of the linear program, instead of the logarithm of the value of the program.}:
 \begin{align*}
\text{minimize} \sum_{z,A}w_{z,A}\cdot2^{|A|} \hspace{1cm} \text{subject to} \quad & \forall x : \sum_{A \ni x} w_{f(x),A} \geq 1-\epsilon,\\
& \forall x : \sum_{z,A \ni x} w_{z,A} = 1,\\
&\forall z,A: w_{z,A} \geq 0.
 \end{align*}
\end{definition}
Jain and Klauck showed that the partition bound bounds randomized query-complexity from below. They also showed that randomized query complexity is bounded above by the third power of the partition bound.
\begin{theorem}[\cite{Jain_2010}, Theorem 3]
\label{prt_ub:JK} ~
\begin{enumerate}
\item \label{ptr_lower} $\R_\epsilon(f) \geq \frac{1}{2} \prt_\epsilon (f)$.
\item \label{ptr_cube} $\R_{1/3}(f) \leq \D(f)=O(\prt_{1/3}(f)^3) $.
\end{enumerate}
\end{theorem}
The best known separation between $\D(f)$ and $\prt(f)$ is quadratic \cite{Ambainis_2016}.
Theorem~\ref{prtjk} proves that this is tight for product distributions. As stated in Section~\ref{introduction}, Theorem~\ref{prtjk} improves upon the result of Jain et al. by a polylogarithmic factor.

Jain and Klauck showed that the partition bound is bounded below by the block sensitivity.
\begin{theorem}[\cite{Jain_2010}, Theorem 3]
\label{fbs&prt}
For any error parameter $\epsilon \in [0,1/2)$,
\[\prt_{\epsilon/4}(f) \geq \epsilon \cdot \bs(f) + \log \epsilon - 2.\]
\end{theorem}
We show that the minimum query corruption bound lower bounds the partition bound (see Appendix \ref{app1} for the proof). Our proof closely follows the proof that the corruption bound is asymptotically bounded above by square of the partition bound shown in \cite{Jain_2010}.
\begin{lemma}
\label{ceps&prt}
For any error parameter $\epsilon \in [0,1/2)$,
\[\corr^\times_{\min,2\epsilon}(f) \leq \prt_\epsilon(f)\cdot \log(1/\epsilon).\]
\end{lemma}
Theorem~\ref{prtjk} now follows, combining Theorems~\ref{prt_ub:this1},~\ref{fbs&prt} and Lemma \ref{ceps&prt} together.

We conclude by showing that the query corruption bound is a quadratic upper bound on the distributional query complexity.
\begin{theorem} \label{D&corr}Let $\epsilon \in [0,1/2)$ and $\mu$ a product distribution over the inputs. Then
\[\D_{4\epsilon}^\mu(f) = O\left(\corr_\epsilon(f)^2\right).\]
\end{theorem}
The result follows by combining Theorem \ref{prt_ub:this1} with the following lemma (see Appendix \ref{app2} for the proof).
\begin{lemma} \label{fbs-corr} For any $\epsilon \in [0, 1)$,
$\fbs(f) \leq \corr_{\epsilon}(f).$
\end{lemma}

\section{Open Problems} \label{open}

\paragraph*{Expectational vs. Fractional Certificate.}
What is the largest separation between the two measures?
Is the upper bound $\EC(f) \leq \FC(f)^{3/2}$ tight?
Any smaller upper bound would improve the $\R(f) \leq \FC(f)^3$ upper bound.
Our attempts in finding a function where $\EC(f)$ is asymptotically larger than $\FC(f)$ so far have been unsuccessful.
As evident by the proof of the quadratic separation between $\EC(f)$ and $\C(f)$, such an example would need to have $\FC^z(f) = o(\C^z(f))$ for both $z \in \{0, 1\}$.
Examples of separations between $\FC(f)$ and $\C(f)$ given in \cite{Aaronson_2008} and \cite{Gilmer_2016} do not satisfy these properties.

\paragraph*{Corruption and Partition Bounds.}
Can the proof of Theorem \ref{prt_ub:this1} be extended to non-product distributions?
The definition of the corruption bound is in some sense a relaxation of the certificate compexity.
Can the argument of $\D(f) \leq \C(f)^2$ be extended to the randomized setting in terms of the corruption bound?

\section*{Acknowledgements.}
This work is supported in part by the Singapore National Research Foundation under NRF RF Award
No. NRF-NRFF2013-13, the Ministry of Education, Singapore under the Research Centres of Excellence programme by the Tier-3 grant Grant \textquotedblleft Random numbers from quantum processes\textquotedblright \ No.\ MOE2012-T3-1-009.

D.G. is partially funded by the grant P202/12/G061 of GA \v CR and by RVO:\ 67985840.
Part of this work was done while D.G. was visiting the Centre for Quantum Technologies at the National University of Singapore.

M.S. is partially funded by the ANR Blanc program under contract ANR-12-BS02-005 (RDAM project).

J.V. is supported by the ERC Advanced Grant MQC.
Part of this work was done while J.V. was an intern at the Centre for Quantum Technologies at the National University of Singapore.

We thank Anurag Anshu for helpful discussions.

\bibliography{bibliography}

\newpage

\appendix

\section{Proof of Lemma \ref{ceps&prt}} \label{app1}

\begin{proof}
Let $c=\prt_\epsilon(f)$. Abusing notation, let $\{w_{z,A}\}_{z,A}$ be a primal feasible point which minimizes the objective. Thus $\sum_{z,A}w_{z,A}\cdot 2^{|A|}=2^c$. We immediately have that,
\[
2^c \geq \sum_{z, A: |A|>c \log (1/\epsilon)}w_{z, A}\cdot 2^{|A|}\geq \frac{2^c}{\epsilon}\sum_{z, A: |A|>c \log (1/\epsilon)}w_{z, A}.
\]
implying,
\begin{align}
\label{bigcube}
\sum_{z,A: |A|>c \log (1/\epsilon)}w_{z,A} \leq \epsilon.
\end{align}
Let $\mu$ be any product probability distribution on $\{0,1\}^n$ (in fact, the proof works for any distribution $\mu$). Without loss of generality, assume that, $\Pr_{x \sim \mu}[f(x)=1] \geq \Pr_{x \sim \mu}[f(x)=0]$. We shall show that $\corr_\epsilon^{1,\mu}(f) = O(c)$. That will prove the theorem.

If $\Pr_{x \sim \mu}[f(x)=0]=0$ then $\{0,1\}^n$ is a $0$-error $1$-certificate of co-dimension $0$, and we are done. From now on, we will assume that $\Pr_{x \sim \mu}[f(x)=0]>0$.

Equation~(\ref{bigcube}) and the two primal constraints imply that for each $x \in \{0,1\}^n$,
\begin{align}
&\sum_{A \ni x,|A| \leq c \log (1/\epsilon)} w_{f(x),A} \geq 1-2\epsilon; \label{one1}\\
&\sum_{A \ni x, |A| \leq c \log (1/\epsilon)} w_{1-f(x),A} \leq 2\epsilon. \label{two2}
\end{align}
Multiplying Equations~(\ref{one1}) and~(\ref{two2}) by $\mu_x$, adding the former over $f^{-1}(1)$ and the later over $f^{-1}(0)$, and re-arranging the order of summations we have,
\begin{align}
&\sum_{A: |A| \leq c \log (1/\epsilon)} \sum_{x \in A, f(x)=1} \mu(x) \cdot w_{1,A} \geq (1-2\epsilon) \cdot \sum_{x \in f^{-1}(1)}\mu(x); \label{one'}\\
&\sum_{A : |A| \leq c \log (1/\epsilon)} \sum_{x \in A, f(x)=0} \mu(x) \cdot w_{1,A} \leq 2\epsilon \cdot \sum_{x \in f^{-1}(0)}\mu(x). \label{two'}
\end{align}
Dividing Equation~(\ref{one'}) by Equation~(\ref{two'}) (note that $\sum_{x \in f^{-1}(0)}\mu_x \neq 0$ by our assumption about $\mu$), we have that,
\begin{align}
\frac{\sum_{A: |A| \leq c \log (1/\epsilon)} w_{1,A}\cdot \left(\sum_{x \in A, f(x)=1} \mu(x)\right)}{\sum_{A:  |A| \leq c \log (1/\epsilon)} w_{1,A}\cdot \left(\sum_{x \in A, f(x)=0} \mu(x)\right)} \geq \frac{1-2\epsilon}{2\epsilon} \cdot \frac{\sum_{x \in f^{-1}(1)}\mu(x)}{\sum_{x \in f^{-1}(0)}\mu(x)} \geq \frac{1-2\epsilon}{2\epsilon}. \nonumber
\end{align}
The last inequality above holds because of our assumption about $\mu$.
This implies that there exists a subcube $A$ with co-dimension $|A| \leq c \log(1/\epsilon)$ such that,
\[\frac{\sum_{x \in A, f(x)=1} \mu(x)}{\sum_{x \in A, f(x)=0} \mu(x)} \geq \frac{1-2\epsilon}{2\epsilon}. \]
Thus,
\[\Pr_{x \sim \mu}[f(x)=1 \mid x \in A] \geq 1-2\epsilon.\]
In other words, $A$ is a $2\epsilon$-error  $1$-certificate under $\mu$. We have,
\[\corr_{\min,2\epsilon}^{\mu}(f) \leq \corr_{2\epsilon}^{1,\mu}(f) \leq |A| \leq \prt_\epsilon (f) \cdot \log(1/\epsilon).\qedhere\]
\end{proof}

\section{Proof of Lemma \ref{fbs-corr}} \label{app2}

\begin{proof}
Let $x$ be such that $\fbs(f, x) = \fbs(f)$, and let $b = f(x)$.
We construct a distribution $\mu$ such that $\corr_{\epsilon}^{b,\mu}(f) \geq \fbs(f)$.

Suppose that $x$ has $k$ sensitive blocks $B_1, \ldots, B_k$.
Let $u_1, \ldots, u_k$ be the corresponding solution to the $\fbs(f, x)$ linear program.
Let $c \in (0, 1-\epsilon)$ be a constant and define $\mu(x) = c$ and $\mu(x^{B_i}) = (1-c)\frac{u_i}{\sum_{i=1}^k u_i} = (1-c)\frac{u_i}{\fbs(f)}$.
Clearly, $\mu$ is a probability distribution on $\{0, 1\}^n$.

Let $A$ be an $\epsilon$-error $b$-certificate according to $\mu$ and recall that $Q_A$ is the set of variables fixed by $A$.
Any input $x^{B_i}$ is inconsistent with $A$ iff $B_i \cap Q_A \neq \varnothing$, thus
$$\sum_{i : B_i \cap Q_A \neq \varnothing} \mu(x^{B_i}) = \Pr_{y \sim \mu}\left[ f(y) \neq b, y \notin A\right].$$
We also have
$$\frac{\Pr_{y \sim \mu}\left[ f(y) = b, y \in A\right]}{\Pr_{y \sim \mu}\left[ f(y) = b, y \in A\right]+\Pr_{y \sim \mu}\left[ f(y) \neq b, y \in A\right]} \geq 1-\epsilon$$
by definition of $A$.
Since $\Pr_{y \sim \mu}\left[ f(y) = b, y \in A\right] = c$, this implies
$$\Pr_{y \sim \mu}\left[ f(y) \neq b, y \in A\right] \leq c \cdot \frac{\epsilon}{1-\epsilon}.$$
Then we get
$$\Pr_{y \sim \mu}\left[ f(y) \neq b, y \notin A \right] = \Pr_{y \sim \mu}\left[ f(y) \neq b \right] - \Pr_{y \sim \mu}\left[ f(y) \neq b, y \in A \right] \geq (1-c) - c \cdot \frac{\epsilon}{1-\epsilon} = 1- c \cdot \frac{1}{1-\epsilon}.$$
On the other hand, since $\sum_{i : j \in B_i} u_i \leq 1$ for each $j \in [n]$, we have
$$\sum_{i : B_i \cap Q_A \neq \varnothing} \mu(x^{B_i}) \leq \sum_{j \in Q_A} \sum_{i : j \in B_i} \mu(x^{B_i}) = \sum_{j \in Q_A} \sum_{i : j \in B_i} (1-c)\frac{u_i}
{\fbs(f)} \leq (1-c) \frac{|A|}{\fbs(f)}.$$
Therefore,
$$\frac{\corr_{\epsilon}(f)}{\fbs(f)} \geq \frac{\corr^{b, \mu}_{\epsilon}(f)}{\fbs(f)} \geq \frac{|A|}{\fbs(f)} \geq \frac{1-\epsilon-c}{(1-\epsilon)(1-c)} = \frac{1-\epsilon-c}{1-\epsilon-c+\epsilon c}.\label{corrgtfb}$$
Since the above relation is true for every $c$, we have,
$$\frac{\corr_{\epsilon}(f)}{\fbs(f)} \geq \lim_{c \rightarrow 0}\frac{1-\epsilon-c}{1-\epsilon-c+\epsilon c}=1.$$
Thus we have $\corr_{\epsilon}(f) \geq \fbs(f)$.
\end{proof}

\end{document}